\documentclass{article}
\usepackage{amsfonts}
\usepackage{amsmath}
\usepackage{amssymb}
\usepackage{amsthm}
\usepackage{hyperref}

\usepackage{yaml}

\usepackage{tikz-cd}

\renewcommand{\int}{\operatorname{int}}

\newcommand{\term}{\textit}
\newcommand{\tuple}[1]{\left\langle{#1}\right\rangle}

\theoremstyle{theorem}
\newtheorem{theorem}{Theorem}

\theoremstyle{definition}
\newtheorem{definition}[theorem]{Definition}

\begin{document}

\title{Database-Driven Mathematical Inquiry}
\markright{Database-Driven Mathematical Inquiry}
\author{Steven Clontz}%{Steven Clontz}  %DO NOT FILL IN AUTHOR'S NAMES UNTIL YOU RECEIVE YOUR PROVISIONAL ACCEPT LETTER. SUBMISSIONS TO THE MONTHLY ARE DOUBLE BLIND.

\maketitle

\begin{abstract}
Recent advances in computing have changed not only the nature
of mathematical computation, but mathematical proof and inquiry
itself. While artificial intelligence and formalized mathematics
have been the major topics of this conversation, this paper
explores another class of tools for advancing mathematics research:
databases of mathematical objects that enable semantic search.
In addition to defining and exploring examples of these tools,
we illustrate a particular line of research that was inspired
and enabled by one such database.
\end{abstract}

\section{Introduction}

While machines have been used to aid in mathematical computation
for centuries (with perhaps the venerable abacus being the first
such tool, used in Mesopotamia as far back as 2300 BCE
\cite{ifrah2000universal}), the use
of mechanical or electronic tools to advance the business of
mathematical proof is a much more recent phenomenon.

In
his paper \textit{What is the point of computers?
A question for pure mathematicians} \cite{buzzard2022point},
Kevin Buzzard provides an
excellent survey of the history of computationally-aided
mathematical inquiry over the past few decades. While exhaustive
search, as was utilized to prove the classic four-color theorem
from graph theory \cite{gonthier2008four}, is one method of using
computers to tackle difficult problems constrained to a finite
search space, recent years have increased the accessibility
of so-called formal verification of mathematics, the main topic
of Buzzard's survey. Such tools go
under various names: ``proof assistants'', 
``theorem provers'', ``computer-verified mathematics'',
and so on; perhaps the most
active modern example is Lean \cite{deMoura2018}.

Such tools provide a mechanism for mathematicians to formalize
their arguments as computer code
that can be carefully checked by machine down to the finest detail.
Of course, needing to refine down to first principles
the usual hand waving of
muddlesome technical details is a non-starter
for most mathematicians; as such,
these tools provide "tactics" that encode standard arguments
into reusable instructions on how to build up
the rigorous low-level proof. As an example, consider how
one might prove that $(ab+ac)+ad=ba+(ca+da)$ for all $a,b,c,d\in\mathbb R$.
While one might accomplish this through applications of
associativity and commutativity (in Lean's \texttt{mathlib}, one might
\verb|rewrite| this goal using the theorems \verb|add_assoc| and \verb|mul_comm|),
the single Lean tactic \verb|ring| alternatively
handles this proof in just four characters: \textit{of course} the
usual properties of \verb|ring|s show that these expressions
are equal to the other.

While such tactics have enabled the usual hand-waving of
technicalities in many cases, it has not (to date) entirely
brought formalized mathematics up to speed with the pace of
classical mathematics. In \cite{taomachine}, Terence Tao 
suggests that the difficulty ratio of writing a correct formal proof
compared with writing a correct informal proof ``is still well above
one[, but] I believe there is no fundamental obstacle
to dropping this ratio below one, especially with increased
integration with AI [...] and other
tools; this would be transformative to our field.'' Indeed, one
may imagine a not-too-distant future where AI tools can
predict formal mathematical proofs,
which are then immediately checked by a computer verification systems,
sending feedback to the AI, rinse and repeat,
e.g. \cite{yang2023leandojo}.

\section{Semantic Mathematical Databases, Large and Small}

However, there is an intermediate goal suggested by Buzzard in his survey:
databases of mathematical concepts that enable semantic search,
which can be used to query the corpus of peer-reviewed and/or
formalized mathematics.
Buzzard argues that ``this sort of tool [...]
has the potential to beat the techniques currently used by PhD students
(`google hopefully,' `page through a textbook/paper hopefully,'
`ask on a maths website and then wait,' `ask another human') hands down''.

But the benefits of curated databases of mathematical knowledge
are not limited to students.
A memory that has stuck with me since graduate school is that
of a particular topology seminar. 
The invited speaker noted at one point
during his talk (something like), ``and I will buy
a bottle of wine a dinner tonight for
whoever can help me remember how to construct
an example of a normal $T_0$ space which is both arc connected and
ultraconnected, but fails to be regular''. This was the early 2010s,
and many of us had smart phones in our pocket, so surely this would be
easily answered by a Google search? Even today, the ``trivial'' example
(take the reals with the topology of rays $(x,\infty)$,
example S42 of the $\pi$-Base \cite{pibase}) is not
revealed by a standard web search using those terms,
and most of the trailheads that exist today and
might lead one to this answer are websites and forum posts that did
not exist at the time. Indeed, while Google has improved
over the past decade or so in how it guesses semantics of words from
their context in a search (or the user's search history),
one might not be surprised that searching
for the words ``regular'' and ``normal'' at that time
could fail to specifically target mathematical results.

Perhaps
what stuck with me was the seeming consensus of the room that
\textit{of course} such an example existed, even though no one could
quite articulate its details on the fly. As a student at the time,
this was fairly frustrating; while now I'm experienced enough to
understand the phenomenon, this is the kind of knowledge that only
comes with time immersed in a field. And even for experts, understanding
that a result is ``known'' is not as useful as knowing ``where'' and
``how'', should the result or its details be needed for another application.\footnote{
During the preparation of this note, my colleague
Katja Berčič shared this relevant
(but perhaps
apocryphal) anecdote concerning the graph theorists W.T. Tutte and
Harold Coxeter: the latter mathematician asked the former in a phone
call whether he knew of anyone studying a certain beautiful example of a
$3$-regular graph with $28$ vertices and $42$ edges. As the story
goes, Tutte replied that of course he had, and that it was quite
famously known by researchers in his circle as a
\href{https://en.wikipedia.org/wiki/Coxeter_graph}{Coxeter graph}.
}

Let's define a \textbf{semantic mathematical database} as a database of
mathematical objects or concepts that includes some technological
measure to automatically connect its entries based upon their mathematical meaning.
This distinction might be best illustrated by considering a common
non-example: mathematics wikis. In addition to the wealth
of mathematical knowledge curated by Wikipedia's dedicated contributors,
there are several mathematics-specific wikis as well: ProofWiki.org,
Math.Fandom.com, EnyclopediaOfMath.org, nCatLab.org, and more. While
these certainly provide a great service by curating mathematical
knowledge in an open and collaborative fashion, and the concepts
they cover are indeed interconnected by relevant hyperlinks, all
these connections must be made by human contributors. For example,
ProofWiki doesn't ``know'' the difference between the
\href{https://proofwiki.org/wiki/Definition:Rational_Number}{rational numbers}
and the
\href{https://proofwiki.org/wiki/Definition:Irrational_Number}{irrational numbers};
each page is just a human-curated document with human-authored hyperlinks
pointing from one page to the next.

In contrast, while entries in a semantic mathematical database frequently
contain narrative descriptions of the objects they represent, this exposition
must be paired with \textbf{metadata} describing its mathematical features.
An extreme example of this is the \textit{L-functions and modular forms database},
or LMFDB for short \cite{lmfdb}. 
To quote their homepage, ``The LMFDB is a database of mathematical objects arising in number theory and arithmetic geometry that illustrates some of the mathematical connections predicted by the Langlands program.''
Like on
a wiki, each object is given its own homepage
(e.g. \url{https://www.lmfdb.org/L/2/31/31.30/c0/0/0} for the L-function
given the ID 2-31-31.30-c0-0-0),
but unlike a wiki, there is no exposition at all
for the majority of entries. Instead, visitors of 2-31-31.30-c0-0-0
are provided a listing of its metadata: its Dirichlet series
($L(s)=1-2^{-s}-5^{-s}+-7^{-s}+8^{-s}+9^{-s}+\dots$), its functional equation
($\Lambda(s)=31^{s/2}\Gamma_{\mathbb C}(s)L(s)=\Lambda(1-s)$), its degree ($2$),
its conductor ($31$), its sign ($1$), and so on.

This lack of exposition is quite natural, as the LMFDB is an example of a
\textbf{large} semantic mathematical database, wherein the vast majority of
its entries are populated programmatically. For example, a sentence description
is provided for \url{https://www.lmfdb.org/L/1/1/1.1/r0/0/0},
the well-known Riemann zeta function, but to populate the LMFDB with
the other $1,448,483$ Dirichlet L-functions
associated to primitive Dirichlet characters with conductor $N\le 2800$,
an exhaustive computation by algorithms described in David Platt's PhD
thesis was required \cite{platt2011computing}.\footnote{
It's worth clarifying that
there are indeed several curated subsets of the LMFDB that include
exposition, for example,
\href{https://www.lmfdb.org/EllipticCurve/Q/interesting}
{this list of ``interesting'' elliptical curves}.
}

Given the computational intensity of generating so much data, the value of
having such databases pre-computed is immediately seen. And with well-designed
tools for navigation and \textbf{semantic} search (wherein we not only look for
the presence/absence of certain keywords in a name or description,
but are able to filter objects based upon the mathematical properties
stored within their metadata),
researchers and students have the potential to find the needles they need
within such haystacks.

In contrast, we may also consider \textbf{small} semantic mathematical
databases, where each entry is individually contributed and moderated
by hand. While not small in a certain sense (with over $370,000$ entries
to date), the On-Line Encyclopedia of Integer Sequences (OEIS) 
\cite{oeis} is perhaps the best known small mathematics database.
For example, the Fibonacci sequence $0,1,1,2,3,5,8,\dots$
is given the ID A000045 and is available at 
\url{https://oeis.org/A000045}.

While search is supported (A000045 is the first result for the
search \texttt{3,5,8,13,21} on its homepage, despite being
modeled starting at \texttt{0}), one might not consider this
database to be very semantic. For example, while it's noted that
this sequence may be obtained from the formula
\(F(n) = F(n-1) + F(n-2)\) (with $F(0)=0,F(1)=1$),
there's no direct way to obtain
a list of all OEIS sequences obtained from the same formula
but different base values.

In contrast, a small example that is quite semantic is the
House of Graphs (HoG) \cite{coolsaet2023house}. A particularly
unique (not to mention, \textit{slick}) feature of HoG is
its search interface, which allows users to actually draw a graph of
interest with their mouse, and it will query its database for
a match. Of course, users are also able to search by specifying
certain graph invariants; for example,
I was able to obtain over $4,000$ bipartite graphs of
average degree $3$.

However, chief among its
features for me personally (as an only casual graph theorist)
is the following sentence I will paraphrase
from its homepage, which I feel
quite succinctly describes the utility of small semantic databases
for researchers:
\textbf{``Most [mathematicians] will agree that among the vast number
of [objects in a given field] that exist,
there are only a few thousand that can be
considered really \textit{interesting}.''} So if the job of
a large database is to help its users find a needle in a haystack,
then the job of a small database is to serve as a well-curated
museum, featuring the most important or interesting examples
from its domain.

\section{The $\pi$-Base model for small databases}

The $\pi$-Base is another example of a small semantic
mathematical database \cite{pibase}. Its implementation at
\url{https://topology.pi-Base.org} is dedicated to results from
general topology; in the next section we will dig deeper
into some related research. Likewise, the reader is directed
to [redacted] for an article describing its specific history,
along with its cyberinfrastructure
(software, deployment, etc.).

Instead, here we will explore the ``$\pi$-Base model'' for
organizing a small semantic mathematical database, appropriate
for mathematicians interested in
certain \textbf{objects}, the \textbf{properties}
each object may or may not satisfy, and \textbf{theorems}
that relate these properties.

\subsection{Properties}

Properties may be considered the atoms of the $\pi$-Base
model.
Each property is encapsulated in a single text file, illustrated
in Listing \ref{listing:property} for the $T_0$ separation axiom.
Each property is given an ID of the form \verb|P######|,
a short string to serve as its name (along with any alternative
aliases the property may go by), and a longer description with
appropriate references. In this sense, properties are quite
similar to the pages of a wiki (albeit with a different mechanism
of peer review, to be described in \ref{peer-review});
it is how they are integrated into the objects and theorems that
makes this database a semantic one.

\begin{lstlisting}[language=yaml,caption=Source for \href{https://topology.pi-base.org/properties/P000001}{property P1 on $\pi$-Base},label=listing:property]
---
uid: P000001
name: "$T_0$"
aliases:
  - Kolmogorov
  - T0
refs:
  - doi: 10.1007/978-1-4612-6290-9
    name: Counterexamples in Topology
---
\end{lstlisting}
\begin{lstlisting}[basicstyle=\small\ttfamily]
Given any two distinct points, there is an open set containing
one but not the other.

Defined on page 11 of {{doi:10.1007/978-1-4612-6290-9}}.
\end{lstlisting}

\subsection{Objects}

Objects have the same features as properties: IDs (of
the form \verb|S######| at topology.pi-base.org, since it
models topological Spaces), a name with aliases, and
a longer description with references; see Listing
\ref{listing:space}.

\begin{lstlisting}[language=yaml,caption=Source for \href{https://topology.pi-base.org/spaces/S000001}{space S1 on $\pi$-Base},label=listing:space]
---
uid: S000001
name: Discrete topology on a two-point set
aliases:
  - Finite Discrete Topology
refs:
  - doi: 10.1007/978-1-4612-6290-9
    name: Counterexamples in Topology
  - wikipedia: Discrete_space
    name: Discrete Space on Wikipedia
---
\end{lstlisting}
\begin{lstlisting}[basicstyle=\small\ttfamily]
Let $X=2=\{0,1\}$ be the space containing two points and let
all subsets of $X$ be open.

Defined as counterexample #1 ("Finite Discrete Topology") in 
{{doi:10.1007/978-1-4612-6290-9}}.
\end{lstlisting}

So while the $\pi$-Base doesn't ``know'' the mathematical
content of these files, and isn't meant to automatically verify
this content (a task left to its volunteer reviewers, see
\ref{peer-review}),
it is aware of the boolean values asserted for
specific space-property pairs. While space S1 is indeed P1,
we will see soon why this is not manually asserted in the $\pi$-Base
database; we instead illustrate in Listing \ref{listing:trait} how
S1 is asserted to be P52 (discrete).

\begin{lstlisting}[language=yaml,caption=Source for \href{https://topology.pi-base.org/spaces/S000001/properties/P000052}{space-property pair S1$|$P52 on $\pi$-Base},label=listing:trait]
---
space: S000001 # The two-point discrete space
property: P000052 # The discrete property
value: true
refs:
- doi: 10.1007/978-1-4612-6290-9
  name: Counterexamples in Topology
---
\end{lstlisting}
\begin{lstlisting}[basicstyle=\small\ttfamily]
All subsets of this space are open by definition.
\end{lstlisting}

\subsection{Theorems}

So why not assert that the two-point discrete space is
$T_0$? The glue that holds the $\pi$-Base model together is its
theorems. Each theorem asserts that if an object satisfies
boolean values for one or more properties, it must also satisfy
the given boolean value for another property.

Therefore, there's no need to assert that S1 is P1, given 
the assertion it is P52, along with the
theorems in listings \ref{listing:thm1} and \ref{listing:thm2}.

\begin{lstlisting}[language=yaml,caption=Source for \href{https://topology.pi-base.org/theorems/T000042}{theorem T42 on $\pi$-Base},label=listing:thm1]
---
uid: T000042
if:
  P000052: true # The discrete property
then:
  P000002: true # The $T_1$ separation axiom
refs:
- doi: 10.1007/978-1-4612-6290-9
  name: Counterexamples in Topology
---
\end{lstlisting}
\begin{lstlisting}[basicstyle=\small\ttfamily]
Asserted on Figure 9 of {{doi:10.1007/978-1-4612-6290-9}}.
\end{lstlisting}
\begin{lstlisting}[language=yaml,caption=Source for \href{https://topology.pi-base.org/theorems/T000119}{theorem T119 on $\pi$-Base},label=listing:thm2]
---
uid: T000119
if:
  P000002: true # The $T_1$ separation axiom
then:
  P000001: true # The $T_0$ separation axiom
refs:
- doi: 10.1007/978-1-4612-6290-9
  name: Counterexamples in Topology
---
\end{lstlisting}
\begin{lstlisting}[basicstyle=\small\ttfamily]
By definition, see page 11 of {{doi:10.1007/978-1-4612-6290-9}}.
\end{lstlisting}

In fact, only three properties are currently asserted for
the two-point discrete space S1, and they characterize it uniquely: 
S52 (the space is discrete), P125 (the space has multiple points),
and \textit{not} P175 (the space has three or more points).

In addition to reducing the effort required to enter data
for every space-property pair in the database, these theorems
drastically reduce the room for error. For one, automated processes
use these theorems to reject any submission to the $\pi$-Base
that would violate a theorem. Furthermore, theorems are used
to aid in search: see
\href{https://topology.pi-base.org/spaces?q=Discrete%2B%7E%24T_0%24}
{this result} seeking examples of discrete spaces that fail to
be $T_0$: not only are no spaces found, the $\pi$-Base application
helpfully generates a proof for why they cannot exist,
by chaining its theorems together.

Indeed, it's this \textbf{automated deduction} that
distinguishes the $\pi$-Base model from most other mathematical
database applications used by mathematicians today,
semantic or otherwise. Buzzard cites this critical
feature, enabled by the formal representation of mathematical
\textit{results} (without necessarily requiring formal
verification of their \textit{proofs}), as a powerful mechanism to
enable ``computer assisted learning'', as well as to serve as an
important stopgap in the absence of tooling that, as a
practical matter, can support full formal verification
of all mathematical results by the general audience
of mathematics researchers \cite{buzzard2022point}.

\subsection{Peer review}\label{peer-review}

All the data in the $\pi$-Base is stored in these four
types of text files. So how does a community of contributors
collaborate on these?

Fortunately, such tools already exist: distributed version
control systems such as Git have served software engineers
since the early 2000s. While Git is not nearly as ubiquitous
a tool for mathematicians, in recent years services such
as GitHub have provided convenient ways for users to
engage with these systems with little more effort than
what's required to log into their web browser and start
writing \cite{g4m}.

Thus, the $\pi$-Base uses these well-established and secure
workflows for the business of peer review of database
contributions, with the added bonus that the entire contents
of its database is available as a free and open-source repository
of text files. In particular, any proposed contribution
to the database comes in the form of a ``pull request''
on GitHub. A discussion thread provided by GitHub for each
pull request allows for contributors
and reviewers to discuss the proposal, and ultimately the
request is either ``accepted and merged'', or ``closed''.

GitHub provides two additional forums: a Discussions
board where community members may freely discuss the project,
and an Issues board where community members may suggest specific
changes, without creating a formal pull request that
would actually implement the proposed improvement.
This allows for contributors to help curate the site, simply
by engaging in conversations with other community members.

Finally, similar to Wikipedia,
$\pi$-Base does not aim to be a home for original research.
However, it was quickly discovered that limiting data to what's
explicitly reflected in textbooks and peer-reviewed literature
is quite constraining: so much of ``known mathematics''
is either never explicitly expressed in such manuscripts, or is
presented in a context not immediately expressed within our model of
objects/properties/theorems. As such, forums such as
Math.StackExchange and MathOverflow are frequently used to ask and answer
questions that fill in $\pi$-Base's gaps,
relocating ``original research'' outside of the $\pi$-Base
GitHub organization to venues with more eyes, as well as helping
disseminate the $\pi$-Base resource to their broader audiences.

\subsection{Applications to other fields, and limitations}

In their 2013 \textit{Notices} article \cite{billey2013fingerprint},
Billey and Tenner pose this scenario:
``Suppose that $M$ is a mathematician and that
$M$ has just proved theorem $T$. How is $M$ to
know if her result is truly new, or if $T$ (or
perhaps some equivalent reformulation
of $T$) already exists in the literature?''\footnote{
One might be reminded of Tai's 1994
rediscovery of the trapezoidal rule from calculus, published 
in the peer-reviewed medical journal
\textit{Diabetes Care}\cite{tai1994mathematical}.
}
In fact, the $\pi$-Base software does exactly this, at least for
for theorems of the form ``all spaces with properties $\mathcal P$
also have properties $\mathcal Q$'', and bounded above by the
exhaustiveness of our community's contributions.

Given the almost-categorical abstractness of the $\pi$-Base
model, one might suspect that its underlying
free and open-source software could be used
for a variety of mathematical disciplines. But while there are
over 80 separate mathematical databases are tracked at 
\href{https://mathbases.org}{MathBases.org}, only
Topology.pi-Base.org uses the $\pi$-Base software itself.

Nonetheless, other databases follow our model to
some degree. A particular example is the Database of Ring
Theory (DaRT) \cite{dart}. DaRT tracks two kinds of objects:
rings and modules. Modules then have properties; however,
they also have connections to rings (e.g. module $M_{20}$
is given by $(x+(x,y)^2)$ over the ring $R_{23}=
F_2[x,y]/(x,y)^2$). Likewise, rings have properties,
but for non-commutative rings these may be asymmetric, that
is, hold differently for left or right multiplication. While
$\pi$-Base's software could in theory be enhanced to allow for this flexibility,
or forked to create a branch version of the software with
this flexibility, this has not yet happened.

Another major blocker for the use of $\pi$-Base software
for many databases is the need to not only model 
boolean properties,
but valued properties. As an example in graph theory, consider
regularity. A regular graph is a graph for which every vertex
has the same degree (number of incident edges). However,
many graph theoretic results need consider not only the regularity
of the graph, but the type of regularity: an $n$-regular graph
is a graph for which every vertex has degree $n$ specifically.
One could manually create separate properties for
$1$-regular, $2$-regular, and so on; at least, the amount
needed would be bounded above by the number of different
graphs in the database. But of course, this is not a very
elegant solution; the more sensible approach would perhaps
be to store not only boolean values for object-property pairs, but
also numerical values.

But this then reveals another limitation
of the principle domain served by the $\pi$-Base today:
properties from general topology typically aren't variable in
terms of an integer or even real number value, but instead
infinite cardinals.
As a concrete example, consider the topological property
of cardinality itself. As of writing, there are over twelve
different properties in the $\pi$-Base
that describe the cardinality of a space:

\begin{itemize}
\item $|X|=0$ (P137),
\item $|X|\geq 2$ (P125),
\item $|X|\geq 3$ (P175),
\item $|X|\geq 4$ (P176),
\item $|X|<\aleph_0$ (P78, i.e. finite),
\item $|X|\leq\aleph_0$ (P57, i.e. countable),
\item $|X|=\aleph_0$ (P181, i.e. countably infinite),
\item $|X|=\aleph_1$ (P114, i.e.
the smallest uncountable cardinality)
\item $|X|<\mathfrak c$ (P58, i.e., smaller than the cardinality of the reals)
\item $|X|\leq\mathfrak c$ (P163, i.e. no larger than the cardinality of the reals)
\item $|X|=\mathfrak c$ (P65, i.e. the cardinality of the reals)
\item $|X|\leq 2^{\mathfrak c}$ (P59, i.e. no larger than the cardinality of the power set of the reals)
\end{itemize}

For now, this ugly mess is the best idea we have! Naively, one might
consider a metaproperty of cardinality that could be assigned
a cardinal value. But the $\pi$-Base is implemented in Javascript,
whose built-in \texttt{Number} type merely supports a
single non-finite value \texttt{Infinity}.
Furthermore, the relationships of some of these properties depend
on your model of ZFC: the $\pi$-Base does not have any theorem
relating P114 and P58 above, as any such theorem would either 
imply the Continuum Hypothesis or its negation! Certainly,
progress could be made here, but it will take the community time
to find the most elegant and sustainable solution, particularly
if the goal is for $\pi$-Base software to support areas besides
general topology.

\section{The $\pi$-Base as a vehicle for original research}

I conclude this manuscript by sharing a tale of how the
author's attempt to model in the $\pi$-Base certain ``well-known''
(scare quotes quite intentional) results led to new mathematics
research. A more technical article covering the results referenced here
is available as [redacted]; here I instead focus on the story that
led to these results.

The first iteration of the $\pi$-Base was a digital representation
of Steen and Seebach's \textit{Counterexamples in Topology} 
\cite{steen1978counterexamples}. In this text, the usual topological
separation axioms\footnote{
It's worth noting that the authors of \textit{Counterexamples}
did not follow the modern convention where $x\geq y$ if and only if
$T_x\Rightarrow T_y$; another advantage of living semantic databases
is to correct and disambiguate questionable notational decisions
of the past.
} $T_0$ through $T_6$ were covered. In particular, let's
focus on these two:

\begin{definition}[P3 of $\pi$-Base]
A space is $T_2$ (or Hausdorff) provided for each pair of distinct
points $x,y$, there exist disjoint open sets $U,V$ with $x\in U$
and $y\in V$.
\end{definition}

\begin{definition}[P2 of $\pi$-Base]
A space is $T_1$ provided for each pair of distinct
points $x,y$, there exist 
(not necessarily disjoint) open sets $U,V$ with $x\in U\setminus V$
and $y\in V\setminus U$.
\end{definition}

Of course, a theorem was created in $\pi$-Base
to assert the immediate-from-the-definitions conclusion that
$T_2\Rightarrow T_1$. But during those early days,
the following two properties were contributed as well:

\begin{definition}[P99 of $\pi$-Base]
    A space is \term{US} (``Unique
    Sequential limits'')
    provided that every
    convergent sequence has a unique limit.
\end{definition}

\begin{definition}[P100 of $\pi$-Base]
    A space is \term{KC} 
    (``Kompacts are Closed'')
    provided that
    its compact subsets are closed.
\end{definition}

At that time, not much justification was required to
contribute a theorem, so it was asserted that
\(T_2\Rightarrow KC\Rightarrow US\Rightarrow T_1\),
based upon a links to a wiki called TopoSpaces. In
2017, I received a small institutional grant to tidy up the 
$\pi$-Base and ensure it aligned with the peer-reviewed
literature, and I was able to justify this by the following
theorem from Wilansky's \textit{American
Mathematcal Monthly} 1967 article entitled
\textit{Between $T_1$ and $T_2$}.

\begin{theorem}[Thm 1 of \cite{MR0208557}]
\[T_2\Rightarrow KC\Rightarrow US\Rightarrow T_1\]
with no arrows reversing.
\end{theorem}

In order for $\pi$-Base to know none of these arrows
can reverse, three counterexamples needed to be provided:
one which was KC-not-$T_2$, one which was US-not-KC, and
one which was $T_1$-not-US. At this stage, rather than
describe to you such counterexamples, or even redirect you
back to my earlier cited survey which exhaustively
accounts each such example, I'll instead remind you
that there's an app for that: Topology.pi-Base.org! What's
important here is that we've finally catalogued all the
properties implied by $T_2$ and implying $T_1$ that are of
interest to researchers.

Ha, of course not. In 2023, Patrick Rabau contributed to
$\pi$-Base the following property defined in M. C. McCord's
1969 \textit{Transactions} paper \textit{Classifying spaces
and infinite symmetric products} \cite{mccord1969classifying}.

\begin{definition}[P143 of $\pi$-Base]
A space is \term{weakly Hausdorff} or
\(wH\) provided
the continuous image of any compact Hausdorff
space into the space is closed.
\end{definition}

In that paper it was observed that this property lies strictly
between $T_2$ and $T_1$ as well. While McCord seemed
unfamiliar with the KC and US properties, in fact, we have the
following.

\begin{theorem}\label{t2chain5}
\[T_2\Rightarrow KC\Rightarrow wH\Rightarrow US\Rightarrow T_1\]
with no arrows reversing.
\end{theorem}

Since inclusion is a continuous map, \(KC\Rightarrow wH\) follows
immediately. To the best of my knowledge, the first result properly
wiring up $wH\Rightarrow US$ was due to Rabau in
\href{https://math.stackexchange.com/questions/4267169/}
{Math.StackExchange 4267169}: essentially, note that a copy
of a converging sequence with its limit in $\mathbb R$
is compact Hausdorff,
so given a converging sequence with its limit in the space,
by $wH$ this set is closed and therefore cannot have a second limit.

\newtheorem*{fakedef}{Definition(?)}
\newtheorem*{fakethm}{Theorem(?)}

It was at this point I decided to poke around the literature
myself and look for any other significant investigations to
properties within this spectrum of $T_1$-not-$T_2$. Indeed,
I found another well-studied weakening of Hausdorff, along
with the following two theorems.

\begin{fakedef}
A space \(X\) is said to be \term{\(k\)-Hausdorff} or \(kH\) provided that
the diagonal \(\Delta_X=\{\langle x,x\rangle:x\in X\}\) is \(k\)-closed in the product topology
on \(X^2\).
\end{fakedef}

\begin{fakethm}[Theorem 2.1 of \cite{MR0374322}]
    \(kH\Rightarrow KC\).
\end{fakethm}

\begin{fakethm}[Proposition 11.2 of \cite{rezkcompactly}]
    \(wH\Rightarrow kH\).
\end{fakethm}

Oh no. We've already established that $KC\Rightarrow wH$,
and that this arrow does not reverse, so what exactly is
going on here?

As it turned there was no mistake made by the authors,
even though they both operated using the same definition
given above for $k$-Hausdorff. However, they were
using different definitions for \textit{$k$-closed}. So
let's now proceed with a bit more care:

\begin{definition}
A subset \(C\) of a space is \term{\(k_1\)-closed} provided for every
compact subset \(K\) of the space, the intersection \(C\cap K\)
is closed in the subspace topology for \(K\).
\end{definition}

\begin{definition}
A subset \(C\) of a space \(X\) is \term{\(k_2\)-closed} provided for every 
compact Hausdorff space \(K\) and continuous map \(f:K\to X\),
the preimage \(f^\leftarrow[C]\) of $C$ is closed in \(K\).
\end{definition}

\begin{definition}
A space \(X\) is said to be \term{\(k_i\)-Hausdorff} or \(k_iH\) provided that
the diagonal \(\Delta_X=\{\langle x,x\rangle:x\in X\}\) is \(k_i\)-closed in the product topology
on \(X^2\).
\end{definition}

With this mindful distinction made, we now
may establish the following theorem.

\begin{theorem}
\[T_2\Rightarrow k_1H\Rightarrow
KC\Rightarrow wH\Rightarrow k_2H
\Rightarrow US\Rightarrow T_1\]
with no arrows reversing.
\end{theorem}
\begin{proof}
The first two arrows followed immediately from
an alternative characterization of $k_1H$ given in
\cite{MR0374322}: each
compact subspace is closed and Hausdorff.
The next arrow was noted in Theorem \ref{t2chain5},
and the fourth $wH\Rightarrow k_2H$ is from
\cite{rezkcompactly}.

With the final arrow $US\Rightarrow T_1$
already established, it remained to
be shown that $k_2H\Rightarrow US$. This is my
personal contribution to the puzzle, first posted
online at [redacted] and now detailed in my article
[redacted]. The idea of this result comes from noting
that the topological space which is simply a sequence converging to
a unique limit, i.e.
$K=\{1/n:n\in\mathbb Z^+\}\cup\{0\}\subseteq\mathbb R$,
is a compact Hausdorff space.
Then given a sequence $a_n$ in a $k_2H$ space $X$
converging to the limits $l_0,l_1$, the map
$f:K\to X^2$ defined by
$f(1/n)=\tuple{a_n,a_n}$ and
$f(0)=\tuple{l_0,l_1}$ is continuous. Then since
$\Delta_X$ is $k_2$-closed, we have
$f^\leftarrow[\Delta_X]\supseteq\{1/n:n\in\mathbb Z^+\}$
closed. It follows that $0\in f^\leftarrow[\Delta_X]$,
$f(0)=\tuple{l_0,l_1}\in\Delta_X$, and thus the limit
$l_0=l_1$ is unique.

These arrows were each shown to not reverse
by considering existing (or lightly modified) counterexamples,
other than one: I needed to construct a space which
is $k_2H$ but not $wH$. This turned out to be what's now
S165 of $\pi$-Base: the one-point compactification
$X=Y\cup\{\infty\}$ of
the Hausdorff ``Arens-Fort space'' $Y=\mathbb Z\cup\{p\}$
(itself S23 of $\pi$-Base). This
followed from a new lemma:
the one-point compactification of any $KC$ space
(and thus any Hausdorff space) is $k_2H$.
But $X$ is not $wH$: by removing
the particular point $p$ of the Arens-Fort space, the remainder
$\mathbb Z\cup\{\infty\}$
is the one-point (Hausdorff) compactification of a countable
discrete space (and thus the continuous image of a compact
Hausdorff space), but is not closed in $X$.
\end{proof}

All in all, I found it fascinating (and not just a little
serendipitous), that the properties $k_1H$, $k_2H$ and $wH$
investigated more or less independently
all fell in line with Wilansky's original spectrum of
$T_1$ weakenings of Hausdorff. One may consider various
natural modifications of these to obtain more intermediate
examples; for example, requiring unique limits of
transfinite sequences, which would live between $k_2H$ and $US$.
And there are indeed several other $T_1$-not-$T_2$ properties
found in the literature, including semi-Hausdorff
($sH$, P169 of $\pi$-Base), locally Hausdorff ($lH$, P84),
and ``has closed retracts'' ($RC$, P101), which do not
fall in line, even among themselves.

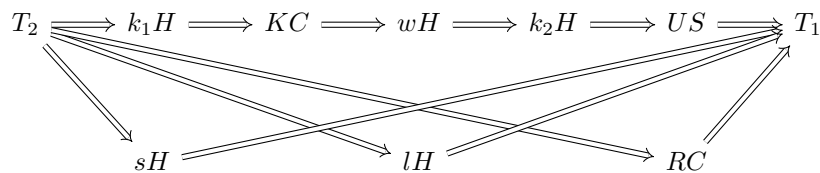
\begin{theorem}
\,

    % https://tikzcd.yichuanshen.de/#N4Igdg9gJgpgziAXAbVABwnAlgFyxMJZABgBpiBdUkANwEMAbAVxiRABUB9AJhAF9S6TLnyEUARnJVajFmwDWncQAl+gkBmx4CRblOr1mrRCADSAYTVCtoogGZ9Mo2wDuqgdZE6UAFkeG5E0Vud3VNLzFkAFZ-WWMQAFUAZSsNYW1IgDZY5xMucVTwjKJJbmkA+IAlSw80m29kBzKDOLY4UM9ilBjmp0CQBndpGCgAc3giUAAzACcIAFskMhAcCCQAdmoGLDB4qAgmACMGVmoACxg6KCQwJgYGahw6LAY2SF3U2YWlx7XESQGOz2B2OpxAFyuNzuDxWz1eJnerFqX0W-1+SD0gI+Jn2RxOIHOl2uiFu90ecLeBCR6hRGPRiAcWOBeLBEOJpJhTxelI+yLmqMZqyQfiZbFxoIJ4KJULJsO5CKpn35wvpMVFOJB+MJkJJ0PJ8vAir530QaqFiGy6pA4q1Up1HP18MNvJpysQmxWf0t22x1s1rOlutlXKdiKVJoAHPTvUCxf7JWyZZyKQqXdM3ctzVGrTaA-a9XLQ0bXSbM38AJxbWMalkJwMOws86npk2Vz1IGO+3N1-PBlPOpEUPhAA
\begin{center}
\begin{tikzcd}
T_2 \arrow[rdd, Rightarrow] \arrow[r, Rightarrow] \arrow[rrrdd, Rightarrow] \arrow[rrrrrdd, Rightarrow] & k_1H \arrow[r, Rightarrow]     & KC \arrow[r, Rightarrow] & wH \arrow[r, Rightarrow]     & k_2H \arrow[r, Rightarrow] & US \arrow[r, Rightarrow]   & T_1 \\
                                                                                                        &                                &                          &                              &                            &                            &     \\
                                                                                                        & sH \arrow[rrrrruu, Rightarrow] &                          & lH \arrow[rrruu, Rightarrow] &                            & RC \arrow[ruu, Rightarrow] &    
\end{tikzcd}
\end{center}
with no arrows reversing or missing.\label{thmdiagram}
\end{theorem}

I again encourage the reader to use the $\pi$-Base application
to obtain any desired details for the proof of
Theorem \ref{thmdiagram}.
I hope you enjoyed this recounting of my investigation,
and perhaps are now primed to agree with the following claim. 

\textbf{Semantic databases are a treasure trove of not only ``well-known''}
(and now, easily accessible to any interested student or researcher
of mathematics)
\textbf{results from the literature, but also an excellent vehicle for
driving future mathematical inquiry.} As of writing,
there are six $T_1$ spaces on the $\pi$-Base for which the
application cannot deduce the semi-Hausdorff property. Unless someone
gets to it first, this would be a perfect graduate student
project, and this is
not the only piece of low-hanging fruit on the
$\pi$-Base. But these things can be subtle,
and simple-sounding problems in general topology
frequently end up yielding numerous questions and
answers at the research level \cite{pearl2011open}.

So I look forward to welcoming more mathematicians to
our community and others that maintain semantic
mathematical databases. However, a major limitation is that
while significant mathematical labor is involved in
contributing to a semantic mathematical database, it
generally doesn't ``count'' when it comes time for one
to enter the job market, or submit one's annual faculty activity
report. To that point, I'll admit that before tenure,
I mindfully avoided spending too much time on $\pi$-Base in order
to focus on my traditional academic output, despite my strong
personal belief in the mathematical value of the service.
But I'll also note that after tenure, when I 
decided to take my work
on $\pi$-Base more seriously, several investigations such as
the one I just described have naturally kept me at least as active
in authoring traditional publications as I was before.
And with recognition from established mathematics organizations
such as the American Institute of Mathematics \cite{AIMath},
the visibility that comes with leading researchers such
as Kevin Buzzard and Terrence Tao promoting machine-assisted mathematics,
and advocacy from newer organizations such as
\texttt{\href{https://code4math.org}{code4math.org}},
I'm optimistic that not only will the experience of mathematics
research change dramatically in the coming decade, but as a
community we will find ways to appropriately value the academic labor required
to develop and maintain the shared digital infrastructure our
research increasingly relies upon.

\bibliographystyle{vancouver}
\bibliography{sample}

\begin{thebibliography}{10}

\bibitem{ifrah2000universal}
Ifrah G.
\newblock The universal history of computing: From the abacus to quantum computing.
\newblock John Wiley \& Sons, Inc.; 2000.

\bibitem{buzzard2022point}
Buzzard K. What is the point of computers? {A} question for pure mathematicians; 2022.

\bibitem{gonthier2008four}
Gonthier G.
\newblock The four colour theorem: Engineering of a formal proof.
\newblock In: Computer Mathematics: 8th Asian Symposium, ASCM 2007, Singapore, December 15-17, 2007. Revised and Invited Papers. Springer; 2008. p. 333-3.

\bibitem{deMoura2018}
de~Moura L, Kong S, Avigad J, van Doorn F, von Raumer J.
\newblock The Lean Theorem Prover (System Description).
\newblock In: Felty AP, Middeldorp A, editors. Automated Deduction - CADE-25. Cham: Springer International Publishing; 2015. p. 378-88.

\bibitem{taomachine}
Tao T.
\newblock Machine assisted proof.
\newblock Notices of the American Mathematical Society, to appear.

\bibitem{yang2023leandojo}
Yang K, Swope A, Gu A, Chalamala R, Song P, Yu S, et~al.
\newblock {LeanDojo}: Theorem Proving with Retrieval-Augmented Language Models.
\newblock In: Neural Information Processing Systems (NeurIPS); 2023. .

\bibitem{pibase}
{The $\pi$-Base Community}. $\pi$-Base community database of topological counterexamples;.
\newblock \url{https://topology.pi-base.org}.

\bibitem{lmfdb}
{LMFDB Collaboration} T. The {L}-functions and modular forms database; 2024.
\newblock [Online.
\newblock \url{https://www.lmfdb.org}.

\bibitem{platt2011computing}
Platt DJ. Computing degree 1 L-functions rigorously; 2011.

\bibitem{oeis}
{OEIS Foundation Inc}. The On-Line Encyclopedia of Integer Sequences; 2024.
\newblock [Online.
\newblock \url{https://www.oeis.org}.

\bibitem{coolsaet2023house}
Coolsaet K, D’hondt S, Goedgebeur J.
\newblock House of Graphs 2.0: a database of interesting graphs and more.
\newblock Discrete Applied Mathematics. 2023;325:97-107.
\newblock Available at \url{https://houseofgraphs.org}.

\bibitem{g4m}
Clontz S. GitHub for Mathematicians;.
\newblock Available from: \url{https://g4m.clontz.org}.

\bibitem{billey2013fingerprint}
Billey SC, Tenner BE.
\newblock Fingerprint databases for theorems.
\newblock Notices of the AMS. 2013;60(8):1034-9.

\bibitem{tai1994mathematical}
Tai MM.
\newblock A mathematical model for the determination of total area under glucose tolerance and other metabolic curves.
\newblock Diabetes care. 1994;17(2):152-4.

\bibitem{dart}
Schwiebert RC. The Database of Ring Theory;.
\newblock Available from: \url{https://ringtheory.herokuapp.com}.

\bibitem{steen1978counterexamples}
Steen LA, Seebach JA, Steen LA.
\newblock Counterexamples in topology. vol.~18.
\newblock Springer; 1978.

\bibitem{MR0208557}
Wilansky A.
\newblock Between {$T\sb{1}$} and {$T\sb{2}$}.
\newblock Amer Math Monthly. 1967;74:261-6.
\newblock Available from: \url{https://doi.org/10.2307/2316017}.

\bibitem{mccord1969classifying}
McCord MC.
\newblock Classifying spaces and infinite symmetric products.
\newblock Transactions of the American Mathematical Society. 1969;146:273-98.

\bibitem{MR0374322}
Lawson J, Madison B.
\newblock Quotients of {$k$}-semigroups.
\newblock Semigroup Forum. 1974/75;9(1):1-18.
\newblock Available from: \url{https://doi.org/10.1007/BF02194829}.

\bibitem{rezkcompactly}
Rezk C. Compactly Generated Spaces;.
\newblock URL:https://ncatlab.org/nlab/files/Rezk\_CompactlyGeneratedSpaces.pdf.
\newblock nLab.
\newblock Available from: \url{https://ncatlab.org/nlab/files/Rezk\_CompactlyGeneratedSpaces.pdf}.

\bibitem{pearl2011open}
Pearl EM.
\newblock Open problems in topology II.
\newblock Elsevier; 2011.

\bibitem{AIMath}
AIMath. Workshop on Open-source cyberinfrastructure supporting mathematics research;.
\newblock URL:https://aimath.org/workshops/upcoming/cyberinfrastructure/.
\newblock Available from: \url{https://aimath.org/workshops/upcoming/cyberinfrastructure/}.

\end{thebibliography}

% After you receive your provisional accept letter, you need to re-submit your manuscript with the following information  (commented out in the blind submission) filled in for each author. Just remove the \% comment coding, fill out the information, and re-submit your manuscript.

%\begin{biog}
%\item[Author Name 1] Insert author bio here.
%\begin{affil}
%Department of Mathematics, University America, Washington DC 20036\\
%authorname@ua.edu
%\end{affil}
%\end{biog}

%\begin{biog}
%\item[Author Name 2] Insert author bio here.
%\begin{affil}
%Department of Mathematics, University America, Washington DC 20036\\
%authorname@ua.edu
%\end{affil}
%\end{biog}

\vfill\eject

\end{document}